\documentclass[runningheads]{llncs}

\usepackage{amstext,amsmath}
\usepackage{amssymb}

\usepackage{hyperref}
\usepackage{xspace}
\usepackage{color}
\usepackage{graphics,datetime}
\graphicspath{{fig/}}

\usepackage{todonotes}

\usepackage{vmargin}
  \setmarginsrb{3.7cm}{1.2cm}{3.7cm}{3.cm}{1.0cm}{1.6cm}{0.8cm}{1.4cm}

\usepackage{dsfont}

\usepackage{cite}
\usepackage{fancyhdr}
\usepackage{boxedminipage}
\usepackage{colortbl}

\usepackage[vlined]{algorithm2e}

\renewcommand{\leq}{\leqslant}
\renewcommand{\geq}{\geqslant}
\renewcommand{\epsilon}{\varepsilon}

\newcommand{\bull}{{\mbox{bull}}}

\newcommand{\php}{proper homogeneous pair}
\newcommand{\mphp}{minimally-sided proper homogeneous pair}

\renewenvironment{proof}[1][]{\par \noindent {\bf Proof:#1}\ }{\hfill$\Box$\\}
\newenvironment{proofETH}[1][]{\par \noindent {\bf Proof of Theorem~\ref{thm:AsymptOptimal}:#1}\ }{\hfill$\Box$\\}

\newtheorem{claimN}{Claim}

\newcommand{\probl}[3]{
\begin{flushleft}
\fbox{
\begin{minipage}{12.6cm}
\noindent {\sc #1}\\
          {\bf Input:} #2\\
          {\bf Question:} #3
\end{minipage}}
\medskip
\end{flushleft}
}
\newcommand{\probls}[4]{
\begin{flushleft}
\fbox{
\begin{minipage}{12.6cm}
\noindent {\textsc {#1}}\\
          {\bf Input:} #2\\
          {\bf Parameter:} #4\\
          {\bf Question:} #3
\end{minipage}}
\medskip
\end{flushleft}
}

\newcommand{\poly}{\mathop{\rm poly}}

\author{Henri Perret du Cray \and Ignasi Sau}

\title{Improved FPT algorithms for\\ weighted independent set in bull-free graphs
\thanks{Research supported by the Languedoc-Roussillon Project ``Chercheur d'avenir'' KERNEL.}
}
\titlerunning{Improved FPT algorithms for weighted independent set in bull-free graphs}

\authorrunning{Henri Perret du Cray and  Ignasi Sau}

\institute{AlGCo project-team, CNRS, LIRMM, Montpellier, France.\\ \vspace{.1cm} \email{ henri.perretducray@gmail.com, ignasi.sau@lirmm.fr}\\ \vspace{.3cm}}

\begin{document}

\maketitle
\setcounter{footnote}{0}

\vspace{-.2cm}
\begin{abstract}
Very recently, Thomass\'{e}, Trotignon and Vuskovic [WG 2014] have given an FPT algorithm for \textsc{Weighted Independent Set} in  bull-free graphs parameterized by the weight of the solution, running in time  $2^{O(k^5)} \cdot n^9$. In this article we improve this running time to $2^{O(k^2)} \cdot n^7$. As a byproduct, we also improve the previous Turing-kernel for this problem from $O(k^5)$ to $O(k^2)$. Furthermore, for the subclass of bull-free graphs without holes of length at most $2p-1$ for $p \geq 3$, we speed up the running time to $2^{O(k \cdot k^{\frac{1}{p-1}})} \cdot n^7$. As $p$ grows, this running time is asymptotically tight in terms of $k$, since we prove that for each integer $p \geq 3$, \textsc{Weighted Independent Set} cannot be solved in time $2^{o(k)} \cdot n^{O(1)}$  in the class of $\{\bull,C_4,\ldots,C_{2p-1}\}$-free graphs  unless the ETH fails.

%
%

\vspace{0.25cm} \textbf{Keywords:} parameterized complexity, FPT algorithm, bull-free graphs, independent set, Turing-kernel.
\end{abstract}

\section{Introduction}
\label{sec:intro}

\vspace{-.15cm}
\paragraph{\textbf{\emph{Motivation}}.} Parameterized complexity deals with problems whose instances $I$ come equipped with an additional integer parameter $k$, and the objective is to obtain algorithms whose running time is of the form $f(k) \cdot \poly(|I|)$, where $f$ is some computable function (see~\cite{FlGr06,DF99,Nie06} for an introduction to the field). Such algorithms are called \emph{Fixed-Parameter Tractable} (FPT). A fundamental notion in parameterized complexity is that of \emph{kernelization}, which asks for the existence of polynomial-time preprocessing algorithms that produce equivalent instances whose size depends exclusively (preferably polynomially)  on $k$. We will be only concerned with problems defined on graphs.

In order to obtain efficient FPT algorithms, a usual strategy is to focus on a graph class whose members have a well-defined {\sl structure}, which can then be exploited to design algorithms.  This paradigm has been exhaustively used in the last decades to obtain efficient FPT algorithms for graphs that exclude a fixed graph as a {\sl minor}, relying on the structural characterization of this graph class given by Robertson and Seymour in their seminal work~\cite{RobertsonS03a}. Nevertheless, the situation is quite different in graphs that exclude a fixed graph as an {\sl induced subgraph}, for which the design of FPT algorithms is still in an incipient stage. Quite recently, the structural description of \emph{claw-free} graphs given by Chudnovsky and Seymour~\cite{ChudnovskyS05} has triggered the design of FPT algorithms in this graph class~\cite{CyganPPPW11,HermelinML12,HermelinMLW11}. Even more recently, a structural characterization of \emph{bull-free} graphs has been given by Chudnovsky~\cite{Chudnovsky12a,Chudnovsky12}. In this article we focus on this latter graph class.

The \emph{bull} is the graph defined by the set of vertices $\{x_1,x_2,x_3,y,z\}$ and the set of edges $\{x_1x_2,x_2x_3,x_3x_1,x_1y,x_2z\}$ (see Fig.~\ref{fig:bull_figure} for an illustration). For a graph $F$, a graph $G$ is said to be \emph{$F$-free} if $G$ does not contain an induced subgraph isomorphic to $F$. Note that the class of bull-free graphs contains the classes of $P_4$-free and triangle-free graphs, so in particular it contains all bipartite graphs.

\begin{figure}[h!]
	\centering
         \vspace{-.8cm}
		\includegraphics[width=0.34\textwidth]{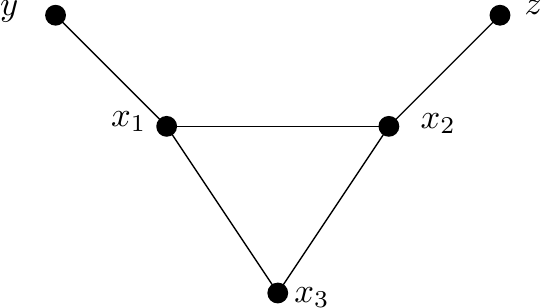}
	\caption{The bull.} \vspace{-.35cm}
	\label{fig:bull_figure}
\end{figure}

An \emph{independent set} in a graph is a set of pairwise non-adjacent vertices. In a vertex-weighted graph, the \emph{weight} of an independent set is the sum of the weights of its vertices. We are interested in the following parameterized problem.

\probls
{Weighted Independent Set}
{An graph $G=(V,E)$ with $|V|=n$, a weight function $w: V \to \mathds{N}$, and a positive integer $k$.}
{Does $G$ contain an independent set of weight at least $k$?}
{The integer $k$.}


The above problem is well-known to be $W[1]$-hard in general graphs~\cite{DF99}, and therefore an FPT algorithm is unlikely to exist (see~\cite{FlGr06,DF99,Nie06} for the missing definitions). Thus, it is relevant to find graph classes for which the problem admits an FPT algorithm, and for which the non-parameterized version still remains NP-hard. In this direction, Dabrowski, Lozin, M{\"u}ller and Rautenbach~\cite{DabrowskiLMR12} gave an FPT algorithm for \textsc{Weighted Independent Set} in $\{\bull, \overline{P_5}\}$-free graphs, where $\overline{P_5}$ is the complement of a path on 5 vertices. Note that the problem is NP-hard in $\{\bull, \overline{P_5}\}$-free graphs, as it is NP-hard in the subclass of triangle-free graphs~\cite{Poljak74}. Recently, Thomass\'{e}, Trotignon and Vuskovic~\cite{ThomasseTV13} generalized this result by giving an FPT algorithm for \textsc{Weighted Independent Set} in the class of bull-free graphs, by exploiting the structural results of Chudnovsky~\cite{Chudnovsky12a,Chudnovsky12}. This article is the starting point of our work, and its main result is the following.

\begin{theorem}[Thomass\'{e}, Trotignon and Vuskovic~\cite{ThomasseTV13}]\label{thm:PreviousFPT} \textsc{Weighted Independent Set} in the class of bull-free graphs can be solved  in time $2^{O(k^5)} \cdot n^9$.
\end{theorem}

\paragraph{\textbf{\emph{Our results}}.} Our main contribution is to improve the running time of the FPT algorithm of Thomass\'{e}, Trotignon and Vuskovic~\cite{ThomasseTV13} stated in Theorem~\ref{thm:PreviousFPT}, specially in terms of the parameter $k$.


\begin{theorem}\label{thm:BetterAlgo} \textsc{Weighted Independent Set} in the class of bull-free graphs can be solved  in time $2^{O(k^2)} \cdot n^7$.
\end{theorem}

We would like to point out that we strongly follow the algorithm of~\cite{ThomasseTV13}, and that our faster algorithm is obtained by improving locally some of the procedures and analyses given in~\cite{ThomasseTV13}. In particular, one of our main improvements relies on a closer look at the structure of the so-called \emph{basic} bull-free graphs as described by Chudnovsky in her series of papers~\cite{Chudnovsky12a,Chudnovsky12}.

It is shown in~\cite[Theorem 7.2]{ThomasseTV13} that the FPT algorithm of Theorem~\ref{thm:PreviousFPT}  actually provides a Turing-kernel\footnote{For a function $g: \mathds{N} \to \mathds{N}$, a parameterized problem $\Pi$ is said to have a \emph{Turing-kernel of size} $g(k)$
if there is an algorithm which, given an input $(I, k)$ together with an oracle for $\Pi$ that decides whether $(I, k) \in \Pi$ in constant time whenever $|I| \leq g(k)$,
decides whether $(I, k) \in \Pi$ in time polynomial in $|I|$ and $k$.} of size $O(k^5)$ for \textsc{Weighted Independent Set}  in bull-free graphs, and that a polynomial kernel  is not possible under reasonable complexity hypothesis. Therefore, as our algorithm follows closely that of Theorem~\ref{thm:PreviousFPT}, from Theorem~\ref{thm:BetterAlgo} we immediately obtain the following corollary.

\begin{corollary}\label{thm:BetterTuring} There exists a Turing-kernel of size $O(k^2)$ for \textsc{Weighted Independent Set}  in the class of bull-free graphs.
\end{corollary}

It is natural to ask whether the algorithm of Theorem~\ref{thm:BetterAlgo} can be improved for subclasses of bull-free graphs. We prove that it is the case when, in addition to the bull, we exclude the holes\footnote{A \emph{hole} in a graph is an induced cycle of length at least 4.} of length at most $2p-1$ for some integer $p \geq 3$ as induced subgraphs. Note that for each $p \geq 3$, the \textsc{Weighted Independent Set} problem is NP-hard in the class of $\{\bull,C_4,\ldots,C_{2p-1}\}$-free graphs, as for each integer $g \geq 3$, its unweighted version is NP-hard in the class of graphs of girth greater than $g$~\cite{Murphy92}, that is in $\{C_3,C_4,\ldots,C_{g}\}$-free graphs, which is a subclass of $\{\bull,C_4,\ldots,C_{g}\}$-free graphs for $g \geq 4$. More precisely, we prove the following theorem.

\begin{theorem}\label{thm:FasterAlgoNoHoles} For each integer $p \geq 3$, \textsc{Weighted Independent Set} in the class of $\{\bull,C_4,\ldots,C_{2p-1}\}$-free graphs can be solved in time $2^{O(k \cdot k^{\frac{1}{p-1}})} \cdot n^7$.
\end{theorem}

In the same way as Corollary~\ref{thm:BetterTuring} follows from Theorem~\ref{thm:BetterAlgo}, from Theorem~\ref{thm:FasterAlgoNoHoles} we obtain the following corollary. It is worth noting that the multipartite construction given in~\cite[Theorem 7.1]{ThomasseTV13} for ruling out the existence of polynomial kernels actually preserves the property of being $\{\bull,C_4,\ldots,C_{2p-1}\}$-free for $p \geq 3$.

\begin{corollary}\label{thm:FasterTuringNoHoles} For each integer $p \geq 3$, there exists a Turing-kernel of size $O(k \cdot k^{\frac{1}{p-1}})$ for \textsc{Weighted Independent Set} in the class of $\{\bull,C_4,\ldots,C_{2p-1}\}$-free graphs.
\end{corollary}

Finally, we provide lower bounds on the running time on any FPT algorithm that solves \textsc{Weighted Independent Set} in the class of $\{\bull,C_4,\ldots,C_{2p-1}\}$-free graphs, for $p \geq 3$. These lower bounds rely on the Exponential Time Hypothesis (ETH), which states that there exists a positive real number $s$ such that {\sc 3-CNF-Sat} with $n$ variables and $m$ clauses cannot be solved in time $2^{sn}\cdot (n+m)^{O(1)}$ (see~\cite{LokshtanovMS11} for more details).

\begin{theorem}\label{thm:AsymptOptimal} For each integer $p \geq 3$, \textsc{Weighted Independent Set} cannot be solved in time $2^{o(k)} \cdot n^{O(1)}$  in the class of $\{\bull,C_4,\ldots,C_{2p-1}\}$-free graphs unless the ETH fails.
\end{theorem}

Note that as $p$ grows, the running time of the algorithm of Theorem~\ref{thm:FasterAlgoNoHoles} tends to $2^{O(k)} \cdot n^7$. As the lower bound given by Theorem~\ref{thm:AsymptOptimal} holds for {\sl any fixed} integer $p \geq 3$, it follows that, as $p$ grows, the running time of the algorithm of Theorem~\ref{thm:FasterAlgoNoHoles} is asymptotically {\sl tight} with respect to the parameter $k$.

\vspace{-.15cm}
\paragraph{\textbf{\emph{Organization of the paper}}.} In Section~\ref{sec:prelim} we state some definitions and results from~\cite{ThomasseTV13} that we need in the remaining sections. Section~\ref{sec:improvedFPT} is devoted to the proof of Theorem~\ref{thm:BetterAlgo}. In Section~\ref{sec:fasterFPT} we focus on bull-free graphs without small holes and prove Theorems~\ref{thm:FasterAlgoNoHoles} and~\ref{thm:AsymptOptimal}. Finally, we conclude with some directions for further research in Section~\ref{sec:conclusions}. Due to space limitations, the proofs of the results marked with `$[\star]$' have been moved to the appendix.

\section{Preliminaries}
\label{sec:prelim}

All the definitions in this section are taken from~\cite{ThomasseTV13}. We use standard graph-theoretic notation (see~\cite{Die05} for any undefined terminology).

\vspace{-.15cm}
\paragraph{\emph{\textbf{Trigraphs}}.} We need to work with \emph{trigraphs} (see~\cite{Chudnovsky12a}), which are a generalization of graphs in which some edges are left ``undecided''. Formally, a trigraph consists of a finite set $V(T)$ of vertices and an adjacency function $\theta: {V(T) \choose 2} \to \{-1,0,1\}$. Two vertices $u,v \in V(T)$ are \emph{strongly adjacent} (resp. \emph{strongly antiadjacent}, resp. \emph{semiadjacent}) if $\theta(uv) = 1$ (resp. $\theta(uv) = -1$, $\theta(uv) = 0$), and in that case $u$ and $v$ constitute a \emph{strong edge} (resp. \emph{strong antiedge}, \emph{switchable pair}). Two vertices $u,v \in V(T)$ are \emph{adjacent} (resp. \emph{antiadjacent}) if $\theta(uv) \in \{0,1\}$ (resp. $\theta(uv) \in \{-1,0\}$), and in that case we say that there is an \emph{edge} (resp. \emph{antiedge}) between $u$ and $v$. Let $\eta(T)$ (resp. $\nu(T)$, $\sigma(T)$) be the set of strongly adjacent (resp. strongly antiadjacent, semiadjacent) pairs of $T$. That is, a trigraph $T$ is a graph if and only if $\sigma(T) = \emptyset$. For a vertex $v \in V(T)$, $N(v)$ (resp.  $\eta(T)$, $\nu(T)$, $\sigma(T)$) denotes the set of vertices in $V(T) \setminus \{v\}$ that are adjacent (resp. strongly adjacent, strongly antiadjacent, semiadjacent) to $v$. The complement $\overline{T}$ of a trigraph $T$ is the trigraph with $V(\overline{T})=V(T)$ and $\theta(\overline{T}) = - \theta(T)$. A trigraph is \emph{monogamous} if every vertex belongs to at most one switchable pair. Most trigraphs considered in this paper will be monogamous.

For two disjoint non-empty subsets of vertices $A,B$ of $V(T)$, we say that $A$ is \emph{strongly complete} (resp. \emph{strongly anticomplete}) to $B$ if every vertex in $A$ is strongly adjacent (resp. strongly antiadjacent) to every vertex in $B$. A \emph{clique} (resp. \emph{strong clique}, \emph{independent set}, \emph{strong independent set}) in $T$ is a set of vertices that are pairwise adjacent (resp. strongly adjacent, antiadjacent, strongly antiadjacent). When we speak about the \textsc{Weighted Independent Set} problem in a trigraph $T$, we are interested in finding an independent set in $T$. We denote by $\alpha(T)$ the maximum weight of an independent set in $T$ (see~\cite{ThomasseTV13} for the precise restrictions of the weight functions defined in trigraphs).

A \emph{realization} of a trigraph $T$ is any trigraph $T'$ such that $\eta(T)\subseteq\eta(T')$,
$\nu(T)\subseteq\nu(T')$, and $\sigma(T') = \emptyset$ (hence $T'$ is a graph). Seen as a trigraph, the \emph{bull} is defined as in Fig.~\ref{fig:bull_figure}, where the corresponding vertices are adjacent or antiadjacent (that is, switchable pairs are allowed). A trigraph is \emph{bull-free} if no induced subtrigraph of it is a bull.

\vspace{-.2cm}
\paragraph{\emph{\textbf{Decomposition of bull-free trigraphs}}.} The algorithm of~\cite{ThomasseTV13}, hence ours as well, is based on a decomposition theorem of bull-free trigraphs that is a simplified version of the one given by  Chudnovsky~\cite{Chudnovsky12a,Chudnovsky12}, and that we proceed to state. We first need two more definitions that will play a fundamental role.

A set $X\subseteq V(T) $ is a \emph{homogeneous set} if $1<\left|X\right|<|V(T)|$ and every vertex in $V(T) \setminus X$ is either strongly complete or strongly anticomplete to $X$. Thus, $V(T) \setminus X$ can be partitioned into two (possibly empty) sets $Y$ and $Z$ such that $X$ is strongly complete to $Y$ and strongly anticomplete to $Z$; see Fig.~\ref{homSetPair} for an illustration, where a solid line means that there are all edges, no line means that there are no edges, and a dashed line means that there is no restriction.

\begin{figure}[tbp]
	\centering \vspace{-.75cm}
		\includegraphics[width=2.4cm]{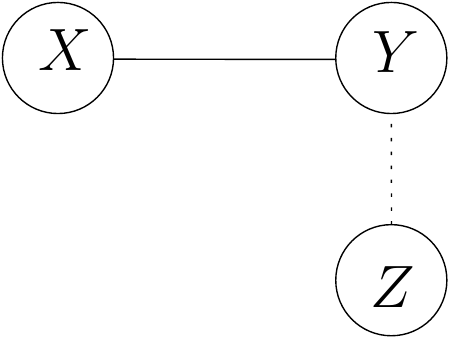}$\ \ \ \ \ \ \ \ \ \ \ \  \ \ \ \ \ \ $\includegraphics[width=3.7cm]{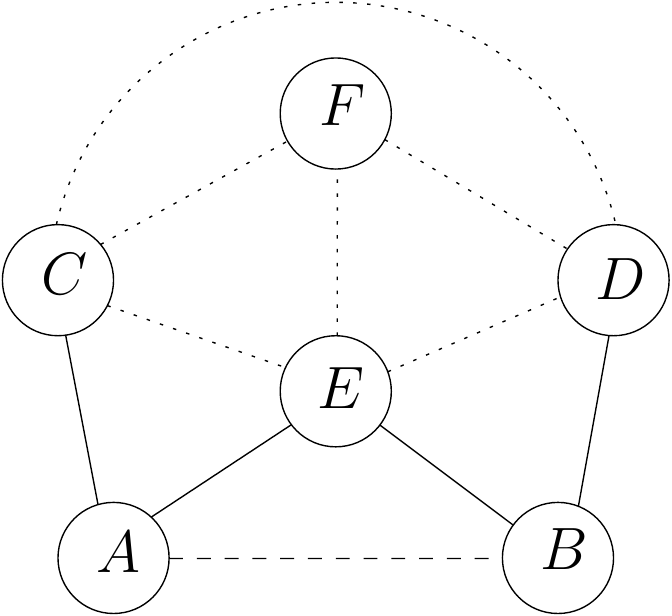}
	\caption{A homogeneous set $X$ and a homogeneous pair $(A,B)$.}
	\label{homSetPair}\vspace{-.25cm}
\end{figure}

A \emph{homogeneous pair} in $T$ is a pair $(A,B)$ of disjoint non-empty subsets of $V(T)$ such that there exist disjoint (possibly empty) subsets $C,D,E,F$ of $V(T)$ such that the following hold:
\begin{itemize}\vspace{-.1cm}
		\item[$\bullet$] $\left\{A,B,C,D,E,F\right\}$ is a partition of $V(T)$;
		\item[$\bullet$] $\left|A\cup B\right|\geq 3$;
		\item[$\bullet$] $\left|C\cup D\cup E\cup F\right|\geq 3$;
		\item[$\bullet$] $A$ is strongly complete to $C\cup E$ and strongly anticomplete to $D\cup F$;
        \item[$\bullet$] $B$ is strongly complete to $D\cup E$ and strongly anticomplete to $C\cup F$; and
         \item[$\bullet$] $A$ is not strongly complete nor strongly anticomplete to $B$.
	\end{itemize}

See again Fig.~\ref{homSetPair} for an illustration. A homogeneous pair is \emph{small} if $\left|A\cup B\right|\leq 6$, and it is \emph{proper} if $C\neq \emptyset$ and $D\neq \emptyset$.

We now define some classes of so-called \emph{basic} trigraphs which will also play an important role in the algorithms. Let $\mathcal{T}_0$ be the class of monogamous trigraphs on at most 8 vertices. Let $\mathcal{T}_1$ be the class of monogamous trigraphs $T$ whose vertex set can be partitioned into (possibly empty) sets $X,K_1,\ldots,K_t$ such that $G[X]$ is triangle free, and $K_1,\ldots,K_t$ are strong cliques that are pairwise anticomplete. According to Chudnovsky's work~\cite{Chudnovsky12a,Chudnovsky12}, the trigraphs in $\mathcal{T}_1$ satisfy some additional conditions that we will detail in Section~\ref{sec:improvedFPT}. This closer look at the class $\mathcal{T}_1$ allows us to significantly improve the dependency on $k$ of the algorithm. Finally, let $\overline{\mathcal{T}}_1 = \{\overline{T}: T \in \mathcal{T}_1\}$. A trigraph is \emph{basic} if it belongs to $\mathcal{T}_0 \cup \mathcal{T}_1 \cup \overline{\mathcal{T}}_1$. We are ready to state the decomposition theorem.

\vspace{-.1cm}

\begin{theorem}[Chudnovsky~\cite{Chudnovsky12a,Chudnovsky12}]
\label{thm:decomposition}
If $T$ is a bull-free monogamous trigraph, then one of the following holds: \vspace{-.1cm}
\begin{itemize}
	\item[$\bullet$] $T$ is basic;
	\item[$\bullet$] $T$ has a homogeneous set;
	\item[$\bullet$] $T$ has a small homogeneous pair; or
	\item[$\bullet$] $T$ has a proper homogeneous pair.
\end{itemize}\vspace{-.15cm}
\end{theorem}

We say that $(X,Y)$ is a \emph{decomposition} of a trigraph $T$ if $(X,Y)$ is a partition of $V(T)$ and either $X$ is a homogeneous cut of $T$ or $X = A \cup B$ where $(A,B)$ is a small or proper homogeneous pair of $T$. A decomposition $(X,Y)$ defines two \emph{blocks} $T_X$ and $T_Y$, whose definition is omitted here, and can be found in~\cite{ThomasseTV13}. A decomposition $(X,Y)$ is a \emph{homogeneous cut} if $X$ is a homogeneous set or $X = A \cup B$ where $(A,B)$ is a proper homogeneous pair. A homogeneous cut $(X,Y)$ is \emph{minimally-sided} if there is no homogeneous cut $(X',Y')$ with $X' \subsetneq X$.

\vspace{-.35cm}
\section{An improved FPT algorithm in bull-free graphs}
\label{sec:improvedFPT}
\vspace{-.15cm}

In this section we give a proof of Theorem~\ref{thm:BetterAlgo}.  We start by providing a high-level description of the FPT algorithm of~\cite{ThomasseTV13}  in Algorithm~\ref{algo:sketch} below (without giving all the details), which will help us to point out the steps for which we provide an improvement.





\vspace{-.4cm}
\begin{algorithm}[h!]
\DontPrintSemicolon
\KwIn{A bull-free trigraph $T$ with $|V(T)|=n$ and the parameter $k$.}
\KwOut{`\textsc{Yes}' if $\alpha(T)\geq k$, and an independent set of weight $\alpha(T)$ otherwise.}
\begin{enumerate}
\item If $T$ is basic, then the problem can be solved in time $O(n^4m) + 2^{O(k^5)}$, where $m$ is the number of strong edges in $T$.
\item Otherwise, by Theorem~\ref{thm:decomposition}, $T$ admits a decomposition. Furthermore, it is shown that $T$ admits a so-called \emph{extreme} decomposition, which is a decomposition $(X,Y)$ such that the block $T_X$ is basic and both $T_X$ and $T_Y$ are bull-free trigraphs. This extreme decomposition can be found in time $O(n^8)$.
	\begin{itemize}	
	\item[2.1.] First, Step~1 is run on the basic bull-free trigraph $T_X$. If $\alpha(T_X) \geq k$, we answer `\textsc{Yes}' and we stop the algorithm. Otherwise, we use the performed computations to build the weighted trigraph $T_Y$.
     \item[2.2.] The whole algorithm is run recursively on the bull-free trigraph $T_Y$.		
	\end{itemize}
\end{enumerate}
\caption{Sketch of the FPT algorithm of~\cite{ThomasseTV13}.}\label{algo:sketch}
\end{algorithm}

As the size of the trigraph $T_Y$ strictly decreases in each recursive step, the overall complexity of Algorithm~\ref{algo:sketch} is easily seen to be upper-bounded by $2^{O(k^{5})}\cdot n^{9}$. (In fact, the algorithm of~\cite{ThomasseTV13} starts by trying to find a decomposition of $T$, and if it fails we know  by Theorem~\ref{thm:decomposition} that $T$ is basic. We reversed the steps in this sketch for the sake of presentation.) Our improvements are the following:
\begin{itemize}
\item[(i)] \textbf{Improvement in terms of the graph size}. We show that in Step~2, an extreme decomposition $(X,Y)$ of $T$ can be found in time $O(n^6)$.
\item[(ii)] \textbf{Improvement in terms of the parameter}. We show that in Step~1, the problem can be solved in basic trigraphs in time $O(n^4m) + 2^{O(k^2)}$.
\end{itemize}

The two improvements above yield the running time given in Theorem~\ref{thm:BetterAlgo}. We now proceed to explain these improvements in detail.

\paragraph{\textbf{\emph{Improvement in terms of the graph size}}.} Our first ingredient  is the following polynomial-time algorithm running in time $O(n^6)$, which should be compared to the algorithm given by~\cite[Theorem 4.3]{ThomasseTV13} that runs in time $O(n^8)$.


\begin{theorem}
\label{thm:4.3} There is an algorithm running in time $O(n^6)$ whose input is a trigraph $T$. The output is a small homogeneous pair of $T$ if some exists. Otherwise, if $G$ has a homogeneous cut, then the output is a minimally-sided homogeneous cut. Otherwise, the output is: ``$T$ has no small homogeneous pair, no proper homogenous pair, and no homogenous set''.
\end{theorem}

The proof of~\cite[Theorem 4.3]{ThomasseTV13} starts by enumerating all sets of vertices of size at most 6 and then it checks whether they define a small homogeneous pair. This procedure takes time $O(n^8)$. Our  first improvement is  a simple algorithm that finds small homogeneous pairs $(A,B)$ in time $O(n^6)$, if there exists one. Without loss of generality, we can assume that $|A|\geq |B|$. The main idea is to fix the vertices of $A$ and then try to find a suitable $B$ verifying $|A\cup B|\leq 6$. While we have not found a small homogeneous pair, we execute Algorithm~\ref{algo:smallHomPair} below for all possible pairs of positive integers $(i,j)$ such that $3\leq i+j \leq 6$ and $j\leq i$ (note that there are at most 8 such pairs), in lexicographic order for $i \in \{2,\ldots,5\}$ and $j \in \{1,\ldots,\min\{1,6-i\}\}$.

\begin{algorithm}[th]
\DontPrintSemicolon
\KwIn{A trigraph $T$ on $n$ vertices, two positive integers $i$ and $j$ such that $3\leq i+j\leq 6$ and $i\geq j$, and such that $T$ does not contain a small homogeneous pair $(A',B')$ with $|A'|=i$ and $|B'| < j$.}
\KwOut{A small homogeneous pair $(A,B)$ with $|A|=i$ and $|B|=j$, if it exists.}
\Begin{
	\ForAll{subsets $A\subseteq V$ of size $i$}{
		\ForAll{subsets $B'\subseteq V \backslash A$ of size $j-1$}{
			$B=B'$, $R=V \backslash (A\cup B')$.\;
			\While{$|B|\neq j$ and $R\neq\emptyset$}{
				pick a new vertex $v\in R$ and remove it from $R$.\;
				\If{$v$ is neither strongly complete nor strongly anticomplete to $A$, or neither strongly complete nor strongly anticomplete to $B$}{
					add $v$ to $B$.}
			}
			\If{$|B|=j$ and all vertices of $V \backslash (A\cup B)$ are either strongly complete or strongly anticomplete to A and either strongly complete or strongly anticomplete to B}{return $(A,B)$.}		
		}
	}
}
\caption{Algorithm for finding a small homogeneous pair of size $i+j$.}\label{algo:smallHomPair}
\end{algorithm}

\begin{lemma}\label{lemma:algo1}
Algorithm~\ref{algo:smallHomPair} is correct and runs in time $O(n^6)$. That is, a small homogeneous pair in a trigraph $T$ can be found in time $O(n^6)$, if it exists.
\end{lemma}
\begin{proof} Suppose that $T$ contains a small homogeneous pair $(A,B)$ such that $|A|=i$ and $|B|=j$, and that $T$ does not contain a small homogeneous pair $(A',B')$ with $|A'|=i$ and $|B'| < j$ (such a pair would have been found in previous iterations). We claim that there exists a vertex $v \in R$ that is neither strongly complete nor strongly anticomplete to $A$, or neither strongly complete nor strongly anticomplete to $B$. Indeed, otherwise $(A,B \setminus \{v\})$ would be a small homogeneous pair, contradicting the conditions of the algorithm. Let $B' = B \setminus \{v\}$. At some point, the algorithm will consider the pair $(A,B')$, and then it will find the corresponding $v$ and check that the found pair is indeed homogeneous. Since $|A| + |B| \leq 6$, these two operations can be done in linear time. Since $i+j-1$ vertices are guessed, the complexity of the algorithm is $O(n^{i+j}) = O(n^6)$, as $i+j\leq 6$. \end{proof}

The second bottleneck in the proof of~\cite[Theorem 4.3]{ThomasseTV13} is a subroutine that finds a minimally-sided proper homogeneous pair, if it exists, in time $O(n^7)$.  We prove the following lemma.

\begin{lemma}$[\star]$
\label{lemma:algo2} There exists an algorithm running in time $O(n^6)$ that finds a minimally-sided homogeneous cut in a trigraph $T$, provided that $T$ has some homogeneous cut.
\end{lemma}

Lemmas~\ref{lemma:algo1} and~\ref{lemma:algo2} together clearly imply Theorem~\ref{thm:4.3}.

\paragraph{\textbf{\emph{Improvement in terms of the parameter}}.} We now focus on the improvement in Step~1 of Algorithm~\ref{algo:sketch}. It is shown in the proof~\cite[Lemma 6.1]{ThomasseTV13} that \textsc{Weighted Independent Set} restricted to the class $\mathcal{T}_1$ admits a kernel of size $O(k^5)$, and this is what gives the function $2^{O(k^5)}$ in the algorithm of Theorem~\ref{thm:PreviousFPT}, as well as the Turing-Kernel of Corollary~\ref{thm:BetterTuring}. In the following we will show that the kernel in the class $\mathcal{T}_1$ can be improved to $f(k) = O(k^2)$, concluding the proof of Theorem~\ref{thm:BetterAlgo} and of Corollary~\ref{thm:BetterTuring}. This improvement is detailed in the following lemma, which should be compared to~\cite[Lemma 6.1]{ThomasseTV13}. More precisely, in~\cite[Lemma 6.1]{ThomasseTV13} the function $f$ is defined as $f(x)=g(x)+(x-1)(\binom{g(x)}{2}+2g(x)+1)$, where $g(x)=\binom{x+1}{2}-1$. We redefine $f$ as $f(x)=5g(x)$, yielding the desired upper bound.

%
%

\begin{lemma}\label{lemma:NEWLemma6.1}
There is an $O(n^{4}m)$-time algorithm with the following specifications.
\begin{itemize}
\item[]\textbf{\emph{Input}:} A weighted monogamous basic trigraph $T$ on $n$ vertices and $m$ strong edges, in which all vertices have weight at least 1 and all switchable pairs have weight at least 2, with no homogeneous set, and a positive integer $k$.
\item[]\textbf{\emph{Output}:} One of the following true statements:
\begin{enumerate}
\item $n\leq f(k)$;
\item the number of maximal independent sets in $T$ is at most $n^3$; or
\item $\alpha(T)\geq k$.
\end{enumerate}
\end{itemize}
\end{lemma}
\begin{proof} The proof follows closely  that of \cite[Lemma 6.1]{ThomasseTV13}. Let $G$ be the realization of $T$ where all switchable pairs are set to ``strong antiedge''. We first check whether $n \leq f(k)$ in constant time. If this is not the case, we apply~\cite[Theorem 5.4]{ThomasseTV13} to $G$, and check whether Output~2 is true. If not, it just remains to prove that Output~3 is a true statement. The running time of the algorithm is $O(n^{4}m)$.

Since $T$ is basic,  there are three cases to consider. Assume first that $T\in\mathcal{T}_{0}$. If $k\geq 2$, then $f(k)> 8 \geq n$, so the algorithm should have given Output~1, a contradiction. Thus, $k \leq 1$, and Output~3 is true. If $T \in \overline{\mathcal{T}}_{1}$, then by \cite[Lemma 5.9]{ThomasseTV13} $T$ has at most $n^3$ maximal independent sets, so the algorithm should have given Output~2, a contradiction.

Thus, necessarily $T \in \mathcal{T}_{1}$. Suppose for contradiction that $\alpha(T)< k$. We consider the decomposition of $T$ into a triangle-free trigraph $X$ and a disjoint union of $t$ strong cliques $K_1,\ldots,K_t$. In contrast to the proof of \cite[Lemma 6.1]{ThomasseTV13}, we will use the following two properties of the class $\mathcal{T}_{1}$, as described by Chudnovsky~\cite{Chudnovsky12a,Chudnovsky12}:
\begin{itemize}
\item[(i)] Each vertex of $X$ has neighbors in at most two distinct cliques.
\item[(ii)] For each clique $K\in\{K_1,\ldots,K_t\}$, with $K=\{v_1,\ldots,v_r\}$, the neighborhood of $K$ in $T$ is a bipartite trigraph, with bipartition $(A,B)$, such that for all $i\in\{1,\ldots,r\}$, $\mathcal{A}_{i+1}\subseteq \mathcal{A}_{i}$ and $\mathcal{B}_{i}\subseteq \mathcal{B}_{i+1}$, where $\mathcal{A}_{i}=A\cap N(v_i)$ and $\mathcal{B}_{i}=B\cap N(v_i)$ (see Fig.~\ref{fig:figure2}).
\end{itemize}

\begin{figure}[h!]
	\centering \vspace{-.45cm}
		\includegraphics[width=0.66\textwidth]{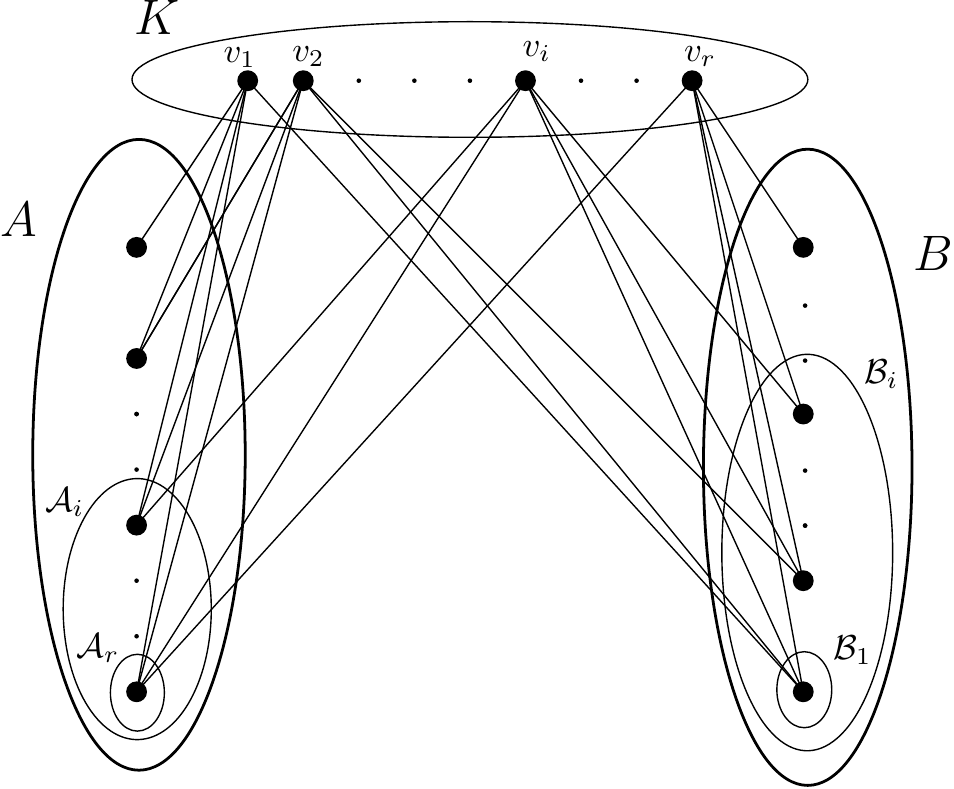}
	\caption{Adjacency between a clique $K$ and the set $X$ in the proof of Lemma~\ref{lemma:NEWLemma6.1}.}
	\label{fig:figure2}
\end{figure}

We can suppose that $|X|\leq g(k)$, otherwise as $T[X]$ is triangle-free, by Ramsey Theorem it follows that $\alpha(G) \geq k$, so we would have that $\alpha(T)\geq \alpha(G) \geq k$.

For $1 \leq i \leq t$, let us denote by $N(K_i)$ the subset of vertices of $X$ that are adjacent to at least one vertex of $K_i$. By Property~(i) above, it holds that
\begin{equation}\label{eq:1}
\sum\limits_{i=1}^{t}{|N(K_i)|}\ \leq\ 2|X|.
\end{equation}

\begin{claimN}\label{claim:1}
For each clique $K\in\{K_1,\ldots,K_t\}$, it holds that $|K|\leq 2|N(K)|$.
\end{claimN}

\begin{proof}
Consider an arbitrary $K\in\{K_1,\ldots,K_t\}$, and let $K=\{v_1,\ldots,v_r\}$. Consider the set $N(K)$ as described by Property~(ii) above. Let us consider $K'=\{v_{i_1},\ldots,v_{i_{r'}}\}$, for $1\leq i_1<i_2<\cdots<i_{r'}\leq r$, the set of vertices in $K$ that do not belong  to any switchable pair. Since $T$ is monogamous, we have that $r-r'\leq|N(K)|$.

Let us note $\mathcal{V}_{j}=\overline{\mathcal{A}}_{i_j}\cup \mathcal{B}_{i_{j}}$, where $\overline{\mathcal{A}}_{i}=A\cap\overline{N(v_i)}$. Note that any two vertices in $K'$ must have a distinct neighborhood, otherwise they form a homogeneous set, a contradiction. Together with Property~(ii), this implies that for all $j\in\{1,\ldots,r'-1\}$, $\mathcal{V}_{j}\subsetneq \mathcal{V}_{j+1}$.

Since $\mathcal{B}_{i}\neq \emptyset$ for all $i\in\{1,\ldots,r\}$, we have that $|\mathcal{V}_{r'}|\geq r'$. And since $\mathcal{V}_{r'}\subseteq N(K')$, we have that $|N(K')|\geq |\mathcal{V}_{r'}| \geq r' =|K'|$.

Therefore, $|K| = r = (r - r ') + r' \leq |N(K)| + |N(K')|\leq 2|N(K)|$, and the claim follows.\end{proof}

Equation~(\ref{eq:1}) and Claim~\ref{claim:1} imply that $\sum\limits_{i=1}^{t}{|K_i|}\leq 4|X|$, and therefore

\vspace{-.15cm}
\begin{equation}\label{eq:2}
|V(T)| \ = \  |X| \ + \ \sum\limits_{i=1}^{t}{|K_i|} \ \leq \ |X| + 4|X| \ = \ 5|X|,
\end{equation}

 that is, $n\leq 5|X|$, and since $|X|\leq g(k)$, the algorithm should have given Output 1, a contradiction.\end{proof}

\vspace{-.45cm}

\section{Independent set in bull-free graphs without small holes}
\label{sec:fasterFPT}

In this section we deal with bull-free graphs without small holes. Namely, we provide a faster FPT algorithm in Subsection~\ref{sec:fasterFPTAlgo} and we prove the lower bound in Subsection~\ref{sec:LowerBound}.

\subsection{Faster FPT algorithm in $\{\bull,C_4,\ldots,C_{2p-1}\}$-free graphs}
\label{sec:fasterFPTAlgo}

In this subsection we prove Theorem~\ref{thm:FasterAlgoNoHoles}. We use the same algorithm described in Section~\ref{sec:improvedFPT} for general bull-free graphs, and the improvement in the time bound for $\{\bull,C_4,\ldots,C_{2p-1}\}$-free graphs consists in a more careful analysis of the kernel size for the basic class $\mathcal{T}_{1}$. More precisely, we will prove that the function $g$ such that $|X| \leq g(k)$ can be redefined as $g_p(x)=x(x^{\frac{1}{p-1}}+2)$. Plugging this function in Equation~(\ref{eq:2}) yields a kernel of size $O(k \cdot k^{\frac{1}{p-1}})$ for the class $\mathcal{T}_{1}$. Indeed, in the proof of Lemma~\ref{lemma:NEWLemma6.1},  if $T$ is a $\{\bull,C_4,\ldots,C_{2p-1}\}$-free trigraph that belongs to the basic class $\mathcal{T}_{1}$, the following lemma implies that in this case it holds that $|X| \leq g_p(k)$, hence proving Theorem~\ref{thm:FasterAlgoNoHoles}. The proof is inspired from classical arguments in Ramsey theory~\cite{Die05} (see also~\cite{Lichiardopol2014} for recent results on the independence number of triangle-free graphs in terms of several parameters).

\vspace{-.15cm}

\begin{lemma}\label{lem:highGirth}
Let $p,k\geq 2$ be two integers and let $G$ be a graph of girth $g(G)\geq 2p$.
If $|V(G)|\geq k(k^{\frac{1}{p-1}}+2)$, then $\alpha(G)\geq k$.
\end{lemma}

\begin{proof}
Let $G'=G$ and $S=\emptyset$. While there exists a vertex $v\in V(G')$ such that $\deg_{G'}(v)<(k^{\frac{1}{p-1}}+1)$, we do the following:
\begin{itemize}
\item[$\bullet$] Add $v$ to $S$; and
\item[$\bullet$] Remove $N[v]$ from $G'$.
\end{itemize}
Note that by construction the set $S$ is an independent set in $G$. When there is no such vertex $v\in V(G')$ anymore, there are two possibilities:
\begin{itemize}
\item[$\bullet$] If $|S|\geq k$, we are done.
\item[$\bullet$] Otherwise, since at each step we removed strictly less than $k(k^{\frac{1}{p-1}}+2)$ vertices from $G$ and by hypothesis $|V(G)|\geq k(k^{\frac{1}{p-1}}+2)$, we have that $V(G')\neq \emptyset$. Note that for all $v\in V(G')$, it holds that $\deg_{G'}(v)\geq(k^{\frac{1}{p-1}}+1)$.
\end{itemize}
In the second case, consider an arbitrary vertex $v\in V(G')$. Let us note $N_i$ the set of vertices at distance $i$ from $v$ in $G'$; see Fig.~\ref{fig:ramsey_gen} for an illustration.
\begin{figure}[h!]
\vspace{-1.0cm}
	\centering
		\includegraphics[width=0.7\textwidth]{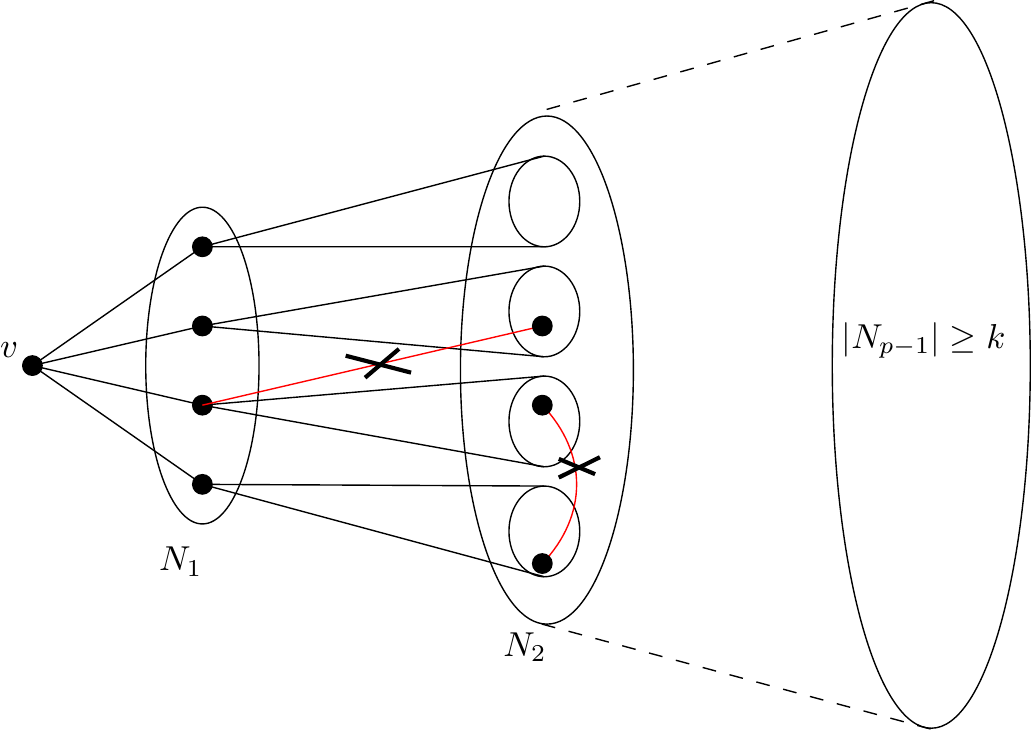}
	\caption{A large independent set in a graph of large girth. The red edges cannot exist.}
	\label{fig:ramsey_gen}
\end{figure}
We shall prove the following two properties by induction for $i \in \{1,\ldots, p-1\}$:
\begin{itemize}
\item[(i)] $N_i$ is an independent set in $G$; and
\item[(ii)] $|N_i|\geq (k^{\frac{1}{p-1}})^{i-1}(k^{\frac{1}{p-1}}+1)$.

\end{itemize}
For $i=1$, $N_1$ is an independent set because $G'$ is triangle-free, as it is an induced subgraph of a graph of girth at least $2p \geq 4$. And we have that $|N_1|=\deg_{G'}(v)\geq k^{\frac{1}{p-1}}+1$.

Suppose that these two properties are true at level $i$, for $1\leq i < p-1$. Let us show that they are also true at level $i+1$. Note first that $N_{i+1}$ is an independent set, as otherwise there would be a cycle in $G'$ of length at most $2i+3\leq 2p-1$, a contradiction (see Fig.~\ref{fig:ramsey_gen}).  On the other hand, two vertices in $N_i$ cannot have a common neighbor in $N_{i+1}$, as otherwise  there would be a cycle in $G'$ of length at most $2i+2\leq 2p-2$, a contradiction (see Fig.~\ref{fig:ramsey_gen}). That is, each vertex in $N_i$ has exactly one neighbor in $N_{i-1}$, and since all vertices in $N_i$ have degree at least $k^{\frac{1}{p-1}}+1$ in $G'$, it follows that
$$
|N_{i+1}|  \ \ \geq\ \ |N_i| \cdot k^{\frac{1}{p-1}}
	 \ \ \geq \ \ (k^{\frac{1}{p-1}})^{i}(k^{\frac{1}{p-1}}+1).
$$
Thus, by induction, $N_{p-1}$ is an independent set in $G$ and $|N_{p-1}|\geq (k^{\frac{1}{p-1}})^{p-2}(k^{\frac{1}{p-1}}+1) \geq k$, as we wanted to prove.\end{proof}

We conclude this subsection with a subtlety that we overlooked so far for the sake of simplicity. In order to have an FPT algorithm for $\{\bull,C_4,\ldots,C_{2p-1}\}$-free graphs, as we claim, we need to make sure that in Algorithm~\ref{algo:sketch} we do not create small holes in the recursive steps. One can check that the block $T_Y$, in which the recursive call is made, does not contain small holes. Nevertheless, the block $T_X$ \emph{may} contain an induced $C_4$ when the switchable pair $\{c,d\}$ is added (see~\cite{ThomasseTV13} for the precise definition of $T_X$). Fortunately,  we can obtain the same asymptotic upper bound of $O(k \cdot k^{\frac{1}{p-1}})$ on the size of $T_X$ in Step~1 of Algorithm~\ref{algo:sketch} when it belongs to the class $\mathcal{T}_{1}$, by using the same arguments, and just distinguishing one more case: if $T_X$ contains a $C_4$, then we apply Lemma~\ref{lem:highGirth} to the graph $T_X \setminus \{c\}$ (or $T_X \setminus \{d\}$), which can be easily seen to be $\{\bull,C_4,\ldots,C_{2p-1}\}$-free, and we just have to add one more vertex ($c$ or $d$) to the upper bound given by Lemma~\ref{lem:highGirth}.


\subsection{A lower bound in $\{\bull,C_4,\ldots,C_{2p-1}\}$-free graphs}
\label{sec:LowerBound}

In this subsection we prove Theorem~\ref{thm:AsymptOptimal}. In fact, we show the lower bound holds even for unweighted  \textsc{Independent Set}. We will reduce from the following problem.

\probl
{Sparse-3-Sat}
{A set of variables $\{x_1,\ldots,x_n\}$ and a set of 3-variable clauses $\{c_1,\ldots,c_m\}$ such that each literal appears at most $c$ times in the clauses, for some constant $c$.}
{Is there an assignment of the variables such that all the clauses are satisfied?}

\vspace{-.15cm}

The \textsc{Sparse-3-Sat} problem cannot be solved in time $2^{o(n)}$ unless the ETH fails (see for instance~\cite{KanjS13}). Our reduction consists of a modification of the classical reduction to show the NP-hardness of \textsc{Independent Set}~\cite{GJ79}.

\vspace{.3cm}

\begin{proofETH}
We will show that if we could solve \textsc{Independent Set} restricted to $\{\bull,C_4,\ldots,C_{2p-1}\}$-free graphs in time $2^{o(k)}\cdot n^{O(1)}$, the we could solve \textsc{Sparse-3-SAT} in time $2^{o(n)}$, which is impossible unless the ETH fails.

We first define a transformation from an instance $\phi$ of \textsc{Sparse-3-Sat} to a graph $G_{\phi}$. With each clause $c_j$, for $1\leq j \leq m$, we associate a triangle where each vertex corresponds to a literal of the clause. For each variable $x\in \{x_1,\ldots,x_n\}$, we add all the edges between the vertices corresponding to $x$ and all the vertices corresponding to $\overline{x}$.


Observe that all the clauses $\phi$ can be satisfied if and only if the graph $G_{\phi}$ has an independent set of size $m$, and that since each literal appears in at most $c$ clauses in $\phi$, the degree of each vertex of $G_{\phi}$ is bounded by $c+2$, hence $|E(G_{\phi})|\leq\frac{3m(c+2)}{2}$.

We now transform the graph $G_{\phi}$ into a $\{\bull,C_4,\ldots,C_{2p-1}\}$-free graph $G_{\phi}'$ by replacing each edge of $G_{\phi}$ with a path on $q$ vertices, where $q$ is the smallest even integer such that $3(q+1)\geq 2p$. See Fig.~\ref{fig:reduction1} for an illustration. The newly added vertices are called \emph{internal}, and the other ones are called \emph{original}.

\begin{figure}[h!]
    \vspace{-.3cm}
	\centering
		\includegraphics[width=0.75\textwidth]{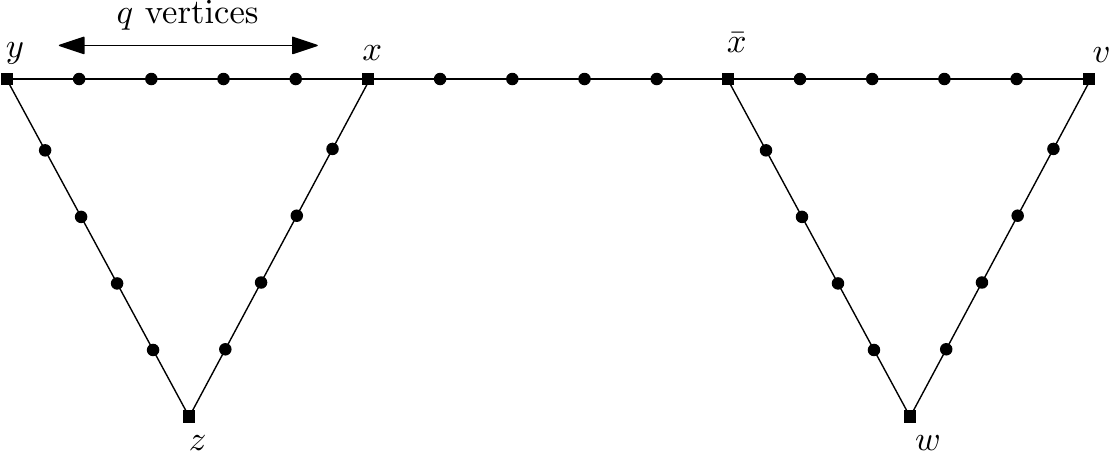}
	\caption{Construction of the graph $G_{\phi}'$ in the proof of Theorem~\ref{thm:AsymptOptimal}.}
	\label{fig:reduction1}
\end{figure}

\begin{claimN}$[\star]$
\label{claim:2} $G_{\phi}$ has an independent set of size $m$ if and only if $G_{\phi}'$ has an independent set of size $|E(G_{\phi})|\cdot\frac{q}{2}+m$. That is, all the clauses $\phi$ can be satisfied if and only if the graph $G_{\phi}'$ has an independent set of size $|E(G_{\phi})|\cdot\frac{q}{2}+m$.
\end{claimN}

To conclude, assume that we can solve \textsc{Independent Set} in $\{\bull,C_4,\ldots,C_{2p-1}\}$-free graphs on $t$ vertices in time $2^{o(k)}\cdot t^{O(1)}$, and let $k = |E(G_{\phi})|\cdot\frac{q}{2}+m$. Then, by Claim~\ref{claim:2}, by solving \textsc{Independent Set} in $G_{\phi}'$ we could solve \textsc{Sparse-3-Sat} in time $2^{o(|E(G_{\phi})|\cdot\frac{q}{2}+m)}\cdot (3m+|E(G_{\phi})|\cdot q)^{O(1)}=2^{o(n)}$, where we have used that  $|E(G_{\phi})|\leq\frac{3m(c+2)}{2}$ and that $m \leq 2c \cdot n$. This is impossible unless the ETH fails.
\end{proofETH}

\section{Conclusions and further research}
\label{sec:conclusions}

We showed in Theorem~\ref{thm:BetterAlgo} that \textsc{Weighted Independent Set} in bull-free graphs can be solved  in time $2^{O(k^2)} \cdot n^7$, and the lower bound of Theorem~\ref{thm:AsymptOptimal} states that the problem cannot be solved in time $2^{o(k)} \cdot n^{O(1)}$  in bull-free graphs unless the ETH fails. Closing this complexity gap (in terms of $k$) is an interesting avenue for further research.

It is tempting to try to apply similar techniques for obtaining FPT algorithms for other (NP-hard) problems in bull-free graphs. The \textsc{Independent Feedback Vertex Set}  problem may be a natural candidate.

Feghali, Abu-Khzam and M\"{u}ller~\cite{FAM14} have recently shown that the problem of deciding whether the vertices of a graph can be partitioned into a triangle-free subgraph and a disjoint union of cliques is NP-complete in planar and perfect graphs. Note that this problem is closely related to deciding whether a given graph belongs to the class $\mathcal{T}_1$ of basic bull-free graphs. Is this problem NP-complete when restricted to bull-free graphs? The recognition of the class $\mathcal{T}_1$ has also been left as an open question in~\cite{ThomasseTV13}.

\bibliographystyle{abbrv}
\bibliography{bib-bull-free}

\newpage
\begin{appendix}

\section{Proof of Lemma~\ref{lemma:algo2}}
\label{ap:algo}

We shall present an algorithm to find a minimally-sided homogeneous cut in a trigraph $T$ that runs in time $O(n^6)$. The algorithm first tries to find a minimally-sided homogeneous set. For doing this, we reuse the same algorithm described in~\cite[Lemma 4.2]{ThomasseTV13}, which runs in time $O(n^2)$. Then, in order to find a minimally-sided proper homogeneous pair, the approach in~\cite{ThomasseTV13} makes $O(n^5)$ calls to the the algorithm of~\cite[Lemma 4.1]{ThomasseTV13}, which runs in time $O(n^2)$, yielding an overall complexity of $O(n^7)$. We proceed to improve this part.

%

We describe in Algorithm~\ref{algo:properHomPair} below how to find minimally-sided proper homogenous pairs. This algorithm is strongly inspired from ~\cite[Lemma 4.1]{ThomasseTV13}, but even if its complexity is still quadratic, the difference lies on the fact that we will need to run Algorithm~\ref{algo:properHomPair} $O(n^4)$ times instead of $O(n^5)$, because we will only need to guess 4 vertices.

More precisely, in order to find a minimally-sided proper homogeneous pair, we run Algorithm~\ref{algo:properHomPair} for all quadruples of vertices $(a_{1},a_{2},c,d)$ such that $a_1$ and $a_2$ are strongly adjacent to $c$ and strongly antiadjacent to $d$.  Therefore, we have an algorithm running in time $O(n^6)$.


We would like to point out that the algorithm does \emph{not} always output a proper homogeneous pair which is \emph{minimally-sided}. Namely, the algorithm outputs the following: either a \php~that may be minimally-sided, or it guarantees that there is no \mphp~$(A,B)$ such that $a_{1},a_{2}\in A$ and $c,d\notin A\cup B$.

For the readability of the algorithm, let $\mathcal{P}$ be the following property:
\vspace{.2cm}

\noindent \textbf{Property~$\mathcal{P}$}: \textit{There is no \mphp~$(A,B)$ such that $a_{1},a_{2}\in A$ and $c,d\notin A\cup B$.}
\vspace{.2cm}

We are now ready to provide a formal description of Algorithm~\ref{algo:properHomPair}.

\begin{algorithm}
\DontPrintSemicolon
\KwIn{A  trigraph $T$, 4 vertices $a_{1},a_{2}, c$, and $d$ such that $a_{1}$ and $a_{2}$ are strongly adjacent to $c$ and strongly antiadjacent to $d$.}
\KwOut{A smallest proper homogeneous pair $(A,B)$ such that $a_{1},a_{2}\in A$ and $c,d\notin A\cup B$, if it exists, or Property~$\mathcal{P}$ otherwise.}
\Begin{
	$R=\{a_{1}, a_{2}\}$, $S=V\backslash R$, $A=\emptyset$, $B=\emptyset$.\;
	We mark the vertices of $V(T)$ as follows:
	\begin{itemize}
	\item[$\bullet$] $\alpha$ for the vertices strongly adjacent to $c$ and strongly antiadjacent to $d$;\;
	\item[$\bullet$] $\beta$ for the vertices strongly adjacent to $d$ and strongly antiadjacent to $c$; and\;
	\item[$\bullet$] $\epsilon$ for the remaining vertices.
	\end{itemize}
	\While{there is a marked vertex $x$ in $R$}{
		\If{$x$ is marked $\epsilon$}{
			Output $\mathcal{P}$.}
		\If{$x$ is marked $\alpha$}{
			Move the following sets from  $S$ to $R$: $\sigma(x)\cap S$, $(\eta(x)\cap S)\backslash\eta(a)$ and $(\eta(a)\cap S)\backslash\eta(x)$. Move $x$ from $R$ to $A$.\;
		}
		\If{$x$ is marked $\beta$}{
			\If{$B$ is empty}{
				Let $b :=x$. Move  $\sigma(b)\cap S$ from $S$ to $R$.\;
				Move $b$ from $R$ to $B$.\;
			}
			\Else{
				Move the following sets from $S$ to $R$: $\sigma(x)\cap S$, $(\eta(x)\cap S)\backslash\eta(b)$ and $(\eta(b)\cap S)\backslash\eta(x)$. Move $x$ from $R$ to $B$.\;
			}
		}
	}
	\If{$B$ is empty} {
		$A$ is a homogeneous set: output $\mathcal{P}$.}
	\Else{
		\If{$B$ is either strongly complete or strongly anticomplete to $A$}{$A$ is a homogeneous set: Output $\mathcal{P}$.}}
		\Else{\If{$|S|\geq 3$}{Output $(A,B)$.}
			\Else{Output $\mathcal{P}$.}
}		
}
\caption{Algorithm for finding minimally-sided proper homogeneous pairs.}\label{algo:properHomPair}
\end{algorithm}

We want to prove that Algorithm~\ref{algo:properHomPair} considers all minimally-sided proper homogeneous pairs, as these pairs are the only ones that may define a minimally-sided homogeneous cut. Let $(A_m,B_m)$ be a minimally-sided proper homogeneous pair, and let $(A_m,B_m,C_m,D_m,E_m,F_m)$ be the corresponding partition. Without loss of generality, we may assume that $|A_m|\geq 2$.


\begin{claimN}\label{claim:3}
The pair $(A_m,B_m)$ is returned by Algorithm~\ref{algo:properHomPair} for a certain quadruple $(a_1,a_2,c,d)$, with $a_1,a_2\in A_m$, $c\in C_m$, and $d\in D_m$.
\end{claimN}
\begin{proof} We proceed to show inductively that by construction, the vertices in $A \cup B$ at the end of the algorithm necessarily belong to all proper homogeneous pairs $(A',B')$ with $a_{1},a_{2}\in A'$ and $c,d\notin A'\cup B'$.

Let $A_i$ and $B_i$ be the sets  $A$ and $B$, respectively, at the end of step $i$ of the algorithm, with $i\leq n$. Let us show that at each step $i$, the sets $A_i$ et $B_i$ satisfy $A_i\subseteq A'$ and $B_i\subseteq B'$.

This property is true for $A_0=\{a_1,a_2\}$ and $B_0=\emptyset$. Suppose it is true at step $i<n$, that is, $A_i\subseteq A'$ and $B_i\subseteq B'$, and let us prove that it is also true at step $i+1$. Let $x_{i+1}$ be the vertex that is added to $A_i$ or to $B_i$ at step $i+1$. As $x_{i+1}\in R$, either $x_{i+1}$ is not strongly adjacent or strongly antiadjacent to $A_i$, or $x_{i+1}$ is not strongly adjacent or strongly antiadjacent to $B_i$. As $A_i\subseteq A'$ and $B_i\subseteq B'$, necessarily $x_{i+1}$ belongs to $A\cup B$. Thus, $x_{i+1}$ is either strongly adjacent to $c$ (if $x_{i+1}\in A$) and then  $x_{i+1}$ is marked $\alpha$ and belongs to $A_{i+1}$, or strongly adjacent to $d$ (if $x_{i+1}\in B$) and then $x_{i+1}$ is marked $\beta$ and belongs to $B_{i+1}$. In both cases, we have that  $A_{i+1}\subseteq A'$ and $B_{i+1}\subseteq B'$.

Therefore,  $A\subseteq A'$ and $B\subseteq B'$, and in particular $A\subseteq A_m$ and  $B\subseteq B_m$. But since $(A_m,B_m)$ is a minimally-sided proper homogeneous set, it follows that $A=A_m$ and $B=B_m$, hence the pair $(A_m,B_m)$ is indeed returned by Algorithm~\ref{algo:properHomPair}.\end{proof}

\section{Proof of Claim~\ref{claim:2}}
\label{ap:claim2}

First, if $G_{\phi}$ has an independent set $S$ of size $m$, we take all the vertices of $S$ and we add $\frac{q}{2}$ internal vertices for each original edge. We can add so many vertices since at most one vertex of each original edge of $G_{\phi}$ can be in $S$.

Conversely, suppose that $G_{\phi}'$ has an independent set $S'$ of size $|E(G_{\phi})|\cdot\frac{q}{2}+m$. As $S'$ cannot contain more than $|E(G_{\phi})|\cdot\frac{q}{2}$ internal vertices, there are at least $m$ vertices of $V(G_{\phi})$ in $S'$.
Let  $\eta$ be the number of edges $xy\in E(G_{\phi})$ such that both $x$ and $y$ are in $S'$. If $\eta=0$, then $S'\cap V(G_{\phi})$ is an independent set of $G$ of size at least $m$, and we are done. We now now that if $\eta>0$, there exists an independent set $S''$ in $G_{\phi}'$ such that $|S''|=|S'|$ and with strictly less than $\eta$ edges $xy\in E(G_{\phi})$ such that both $x$ and $y$ are in $S''$.

Let $x,y \in S'$  be such that $xy\in E(G)$, and let us note $(x=x_0,x_1,\ldots,x_q,y)$ the path between  $x$ and $y$ in $G_{\phi}'$ induced by the subdivision of the edge $xy$. Let $i$ be the smallest integer in $\{1,\ldots ,q\}$ such that $x_i$ and $x_{i+1}$ are not in $S'$. Note that such an integer $i$ exists since $q$ is an even number, and observe that $i$ is an odd number; see Fig.~\ref{fig:reduction2} for an illustration, where the red vertices belong to $S'$.

\begin{figure}[h!]
	\centering
		\includegraphics[width=0.80\textwidth]{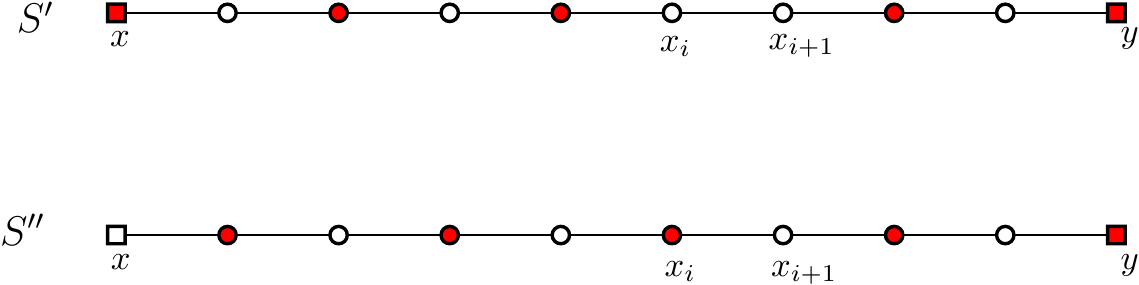}
	\caption{Decreasing the parameter $\eta$ in the proof of Theorem~\ref{thm:AsymptOptimal}.}
	\label{fig:reduction2}
\end{figure}

We now construct $S''$ as follows: we initialize $S''=S'$, and for all $j\in\{0,\ldots, \frac{i-1}{2}\}$, we remove $x_{2j}$ from $S''$ and we add $x_{2j+1}$. Observe that since $x_{i+1}$ is not in $S'$, $S''$ is indeed an independent set of size $|E(G_{\phi})|\cdot\frac{q}{2}+m$ such that the parameter $\eta$ has strictly decreased; see the lower part of Fig.~\ref{fig:reduction2}.

Repeating this procedure while $\eta > 0$, we eventually obtain an independent set $S'_0$ of $G'_{\phi}$ of size $|E(G_{\phi})|\cdot\frac{q}{2}+m$ such that there are no two vertices $x,y\in S'_0$ such that $xy\in E(G_{\phi})$. Therefore, $S'_0\cap V(G_{\phi})$ is an independent set in $G_{\phi}$. Furthermore, it has size at least $m$ since there cannot be more than $|E(G_{\phi})|\cdot\frac{q}{2}$ internal vertices in $S'_0$.

\end{appendix}

\end{document}